\documentclass[11pt]{article}
 
\usepackage[utf8]{inputenc}
\usepackage{authblk}
\usepackage{amsthm}
\usepackage{thm-restate}

\usepackage{fullpage}
\usepackage{soul}
\usepackage{tabularx}
\usepackage{multirow}

\usepackage{amsmath}
\usepackage{amssymb}
\usepackage{amsfonts}
\usepackage{mathtools}
\usepackage{enumitem}
\usepackage{thm-restate}
\usepackage[pagebackref]{hyperref}
\usepackage[sort,numbers]{natbib}
\usepackage[dvipsnames,svgnames,x11names]{xcolor}

\usepackage{tikz}

\usepackage[capitalize]{cleveref}
\usepackage{caption}
\usepackage{xspace}
\usepackage{multicol}

\usepackage{todonotes}
\newcommand{\tw}{{\rm tw}}

\newcommand{\FPT}{$\mathsf{FPT}$ }
\newtheorem{theorem}{Theorem}
\newtheorem{lemma}[theorem]{Lemma}
\newtheorem{observation}[theorem]{Observation}

\newtheorem{proposition}[theorem]{Proposition}

\theoremstyle{definition}
\newtheorem{definition}[theorem]{Definition}

\newtheorem{remark}{Remark}
\crefname{claim}{Claim}{Claims}

\usepackage{xspace}

\let\oldlambda\lambda
\renewcommand{\lambda}{\ensuremath{\oldlambda}\xspace}
\let\oldalpha\alpha
\renewcommand{\alpha}{\ensuremath{\oldalpha}\xspace}
\let\oldDelta\Delta
\renewcommand{\Delta}{\ensuremath{\oldDelta}\xspace}

\newcommand{\eps}{\ensuremath{\varepsilon}\xspace}
\renewcommand{\epsilon}{\eps}

\newcommand{\ignore}[1]{}

\renewcommand{\leq}{\leqslant}
\renewcommand{\geq}{\geqslant}

\title{A Parameterized Perspective on Uniquely Restricted Matchings}

\author[1]{Juhi~Chaudhary\thanks{Supported by the DAE, Government of India, project no.~RTI4001. (email: juhi.chaudhary@tifr.res.in).}}
\author[2]{Ignasi~Sau\thanks{Supported by the French ANR project ELIT (ANR-20-CE48-0008). (email: ignasi.sau@gmail.com).}}
\author[3]{Meirav~Zehavi\thanks{Supported by the ERC, grant no. 101039913. (email: zehavimeirav@gmail.com).}}

\affil[1]{\small School of Technology and Computer Science, Tata Institute of Fundamental Research, Mumbai, 
India.}
\affil[2]{\small LIRMM, Université de Montpellier, CNRS, Montpellier, France.}
\affil[3]{\small Department of Computer Science, Ben-Gurion University of the Negev, Beersheba, Israel.}

\date{}

\begin{document}

\maketitle

\begin{abstract}
Given a graph $G$, a \emph{matching} is a subset of edges of $G$ that do not share an endpoint. A matching $M$ is \emph{uniquely restricted} if the subgraph induced
by the endpoints of the edges of $M$ has exactly one perfect matching. Given a graph $G$ and a positive integer $\ell$, \textsc{Uniquely Restricted Matching} asks whether $G$ has a uniquely restricted matching of size at least $\ell$. In this paper, we study the parameterized complexity of \textsc{Uniquely Restricted Matching} under various parameters. Specifically, we show that \textsc{Uniquely Restricted Matching} admits a fixed-parameter tractable (FPT) algorithm on line graphs when parameterized by the \emph{solution size}. We also establish that the problem is FPT when parameterized by the \emph{treewidth} of the input graph.
Furthermore, we show that \textsc{Uniquely Restricted Matching} does not admit a polynomial kernel with respect to the vertex cover number plus the size of the matching unless $\mathsf{NP}\subseteq \mathsf{coNP} \slash \mathsf{poly}$.
\bigskip

\end{abstract}

\section{Introduction}
Matchings in graphs are among the most fundamental and extensively explored topics in combinatorial optimization. While the classical matching problem can be solved efficiently, introducing additional constraints typically leads to computational intractability. For instance, in a $\mathcal{P}$-matching, the matching $M$ must be such that the subgraph $G[V_M]$, induced by the endpoints of the edges in $M$, satisfies a specified property $\mathcal{P}$. These properties can vary widely, such as requiring the induced subgraph to be itself a matching (\emph{induced matching}), to be connected (\emph{connected matching}), to form a forest (\emph{acyclic matching}), or to contain exactly one perfect matching (\emph{uniquely restricted matching}). In the case of a classical matching, the property $\mathcal{P}$ is simply that $G[V_M]$ forms an arbitrary subgraph of $G$. 

In this paper, we will focus on uniquely restricted matchings. The idea of uniquely restricted matchings was introduced by Golumbic et al.~\cite{golumbic2001uniquely}, inspired by a problem in linear algebra. In bipartite graphs, finding a maximum uniquely restricted matching is closely related to identifying the largest upper triangular submatrix obtainable by permuting the rows and columns of a given matrix. These submatrices are particularly useful because they make solving the corresponding systems of linear equations much faster (see \cite{golumbic2001uniquely} for more details). This concept has interesting applications---not just in matrix theory \cite{hershkowitz1993ranks}, but also in modeling compatibility among a group of agents, such as people, robots, or even spies \cite{golumbic2001uniquely}. The concept of uniquely restricted matching was revisited along with other variants of matching in a paper by Goddard et al., who also discussed the minimum maximal versions of these variants \cite{goddard2005generalized}.

It is worth noting that some of the other matching variants have been extensively studied due to their diverse applications and theoretical significance, while others are more recently explored and gaining attention. The examples include, \textsc{Induced Matching}~\cite{DBLP:conf/stacs/Koana23,moser2009parameterized,cameron1989induced,lampis2025structural,chaudhary2023matchings,song2021induced}, \textsc{Acyclic Matching}~\cite{lampis2025structural,chaudhary2023matchings,furst2019some,panda2012acyclic,chaudhary2025parameterized,goddard2005generalized,panda2023acyclic,hajebi2023parameterized,furst2018lower}, \textsc{Weighted Connected Matching} ~\cite{gomes2025weighted}, and \textsc{Disconnected Matching}~\cite{chaudhary2023matchings, gomes2023disconnected}. The common goal in these problems is to find a $\mathcal{P}$-matching of maximum cardinality. In line with this, in this paper, we focus on the \textsc{Uniquely Restricted Matching} problem, the decision version of which is formally defined as follows:
 \medskip

		\noindent\fbox{ \parbox{160mm}{
		\noindent {\textsc{Uniquely Restricted Matching:}}\\
		\textbf{Input:} An undirected graph $ G$ and a positive integer $\ell$.\\
		\textbf{Question:} Does $G$ have a uniquely restricted matching (urm) of size at least $\ell$?}}
		\medskip

    The \emph{uniquely restricted matching number} of a graph $G$ is
the maximum cardinality of a uniquely restricted matching among all uniquely restricted matchings in $G$, and we denote it by $\mu_{urm}(G)$.
\subsection{Related Work}
Over the years, the computational complexity of \textsc{Uniquely Restricted Matching} has been investigated across a variety of graph classes. \autoref{tablepig1} provides a consolidated overview of these results. Beyond algorithmic aspects, the problem has also been studied from a structural perspective. For instance, in response to a question posed by Golumbic et al.~\cite{golumbic2001uniquely}, Penso et al.~\cite{penso2018graphs} characterized the class of graphs in which some maximum matching is uniquely restricted. They further resolved a problem raised by Levit and Mandrescu~\cite{levit2003local} by identifying the graphs in which every maximum matching is uniquely restricted.
In terms of approximability, Mishra~\cite{mishra2011maximum} showed that \textsc{Uniquely Restricted Matching} is hard to approximate within a factor of $\mathcal{O}(n^{\frac{1}{3}-\epsilon})$ for any $\epsilon > 0$, unless $\mathsf{NP} = \mathsf{ZPP}$, even when the input is restricted to bipartite, split, chordal, or comparability graphs. On the positive side, he provided approximation algorithms with factors of $2 - \frac{4}{n - 2}$ for 3-regular bipartite graphs and 4 for general 3-regular graphs. Building on this, Baste et al.~\cite{baste2019approximating,baste2019upper} improved the approximation ratio to $9/5$ for subcubic bipartite graphs, refining the earlier 2-approximation. In addition to these results, they also analyzed the \emph{uniquely restricted chromatic index}, defined as the minimum number of uniquely restricted matchings into which the edge set of a graph can be partitioned.

\begin{table} [h]
	\scalebox{0.88}{

	\begin{tabular}{|l|l|c|c|}
		\hline

	&	  & \underline{\textbf{Graph Class}} & \underline{\textbf{Result}} \\
 
	1. &	Golumbic et al. \cite{golumbic2001uniquely}  & Bipartite graphs, Split graphs & \textsf{NP}-hard  \\

&	  &  Threshold graphs, Cacti,
Block graphs & linear-time solvable\\

 2.    &  Mishra \cite{mishra2011maximum}  & Bipartite graphs of degree at most 3  & \textsf{APX}-complete\\
     
	3. &	Francis et al. \cite{francis2018uniquely}   & Interval graphs  &  polynomial-time solvable\\

  &	  & Proper interval graphs, Bipartite permutation graphs &  linear-time solvable\\

		\hline
	\end{tabular}

 }
	
		\caption{Overview of algorithmic results for \textsc{Uniquely Restricted Matching}. }
	\label{tablepig1}

\end{table}

While several matching variants have been well-studied within the parameterized complexity framework, the \textsc{Uniquely Restricted Matching} problem remains relatively under-explored in this regard. A recent result establishes that, assuming $\mathsf{W[1]} \nsubseteq \mathsf{FPT}$, there is no $\mathsf{FPT}$ algorithm that can approximate the problem within any constant factor when parameterized by the size of the matching thereby also establishing its $\mathsf{W[1]}$-hardness \cite{chaudhary2025parameterized}.  It is also known that the problem admits a kernel of size $\mathcal{O}(\mu_{\text{urm}}(G))$ on planar graphs, where $\mu_{\text{urm}}(G)$ denotes the size of a maximum uniquely restricted matching in $G$ \cite{chaudhary2025parameterized}.

The minimum maximal version of uniquely restricted matching (Min-Max-URM) has also been investigated in the literature, where the objective is to find a maximal uniquely restricted matching of the smallest possible size. 
The decision version of Min-Max-URM is known to be NP-complete for several graph classes, including bipartite graphs with maximum degree 7, chordal bipartite graphs, star-convex bipartite graphs, chordal graphs, and doubly chordal graphs~\cite{chaudhary2021complexity,panda2017complexity}. On the positive side, Min-Max-URM can be solved in linear time for specific graph classes such as bipartite distance-hereditary graphs, bipartite permutation graphs, proper interval graphs, and threshold graphs~\cite{chaudhary2021complexity}. Furthermore, the problem is APX-complete for graphs with maximum degree 4, and cannot be approximated within a ratio of $n^{1-\epsilon}$ for any $\epsilon>0$ unless $\mathsf{P}=\mathsf{NP}$ even for bipartite graphs \cite{chaudhary2021complexity}.

\subsection{Our Contribution}
In this paper, we study the parameterized complexity of \textsc{Uniquely Restricted Matching} under various parameters, namely solution size, treewidth, and vertex cover number. 

We begin by studying the problem parameterized by the solution size, which is arguably the most natural choice of parameter. Since the problem is W[1]-hard on general graphs under this parameterization, we focus our attention on line graphs, where several variants of matching have already been explored. For acyclic matching, the following characterization is known: Given a graph $H$, its line graph $G=L(H)$ admits an acyclic matching of size $\ell$ if and only if $H$ contains $t$ vertex-disjoint paths $P_1, \ldots, P_t$, for some positive integer $t$, such that each $P_i$ has even length and the total sum of their lengths satisfies $\sum_{i=1}^{t} \text{length}(P_i) = 2\ell$~\cite{hajebi2023parameterized}.  Similarly, for induced matchings, it is known that $L(H)$ has an induced matching of size at least $\ell$ if and only if $H$ has at least $\ell$ vertex-disjoint copies (not necessarily induced) of $P_3$, the path on three vertices~\cite{moser2009parameterized}. 

Motivated by these structural characterizations, we investigate the computational complexity of \textsc{Uniquely Restricted Matching} on line graphs. In particular, we first give the following characterization. 

  \begin{restatable}{lemma}{liner}\label{thm:liner}
     Let $H$ be a graph and $G=L(H)$ be the line graph of $H$. Then, $G$ has a urm of size $\ell$ if and only if $H$ contains $\ell$ edge-disjoint paths $W_1, \ldots, W_\ell$, each of length $2$ such that $\bigcup_{i\in [\ell]}{W_i}$ is a forest, and no two distinct paths $W_i$ and $W_j$ (for $i\neq j$) together form an edge-induced subgraph isomorphic to $K_{1,4}$. 
 \end{restatable}
 
In the above lemma, observe that each $W_i$, for $i \in [\ell]$, corresponds to a matching edge in $L(H)$.
We then prove the following theorem, whose proof, in a nutshell, proceeds as follows: we begin by generating all forests with $2\ell$ edges and then apply a greedy algorithm to filter out the ``relevant" forests. Once we have this collection, we employ the color-coding technique\cite{alon1995color} to efficiently test whether any of these ``relevant" forests appear as a subgraph of $H$, where $G=L(H)$. This entire procedure can be carried out in FPT time.
  \begin{restatable}{theorem}{thmfpt}    
     For every line graph $G$, the \textsc{Uniquely Restricted Matching} problem can be solved in $2^{\mathcal{O}(\ell)} \cdot |V(G)|$ time when parameterized by the solution size $\ell$.
\end{restatable}

Turning to structural parameters, we next study the problem parameterized by the treewidth of the input graph. Note that this parameter has been recently studied for other variants such as \textsc{Induced Matching}, \textsc{Acyclic Matching}, and $c$-\textsc{Disconnected Matching}, where single-exponential time algorithms are obtained~\cite{lampis2025structural,chaudhary2023matchings}. However, for \textsc{Uniquely Restricted Matching}, the problem is more challenging, and we are only able to obtain a super-exponential time algorithm. In particular, we obtain the following algorithmic result.
  \begin{restatable}{theorem} {treewidththm}\label{thm1}
     \textsc{Uniquely Restricted Matching} can be solved in $\mathcal{O}(2^{\tw^2/2}\cdot |V(G)|)$ time by a deterministic algorithm when parameterized by the treewidth $\tw$ of the input graph $G$.
 \end{restatable}

Finally, we consider the parameterization by the vertex cover number. It is straightforward to observe that the problem admits an exponential kernel under this parameter. Indeed, if $(X, I)$ is a partition of $V(G)$ where $X$ is a vertex cover of bounded size, say $t$ and $I$ is an independent set, then we can assume that no two vertices in $I$ have the property that the neighborhood of one is contained in the neighborhood of other. Otherwise, if there exists some $x,y\in I$ such that $N(x)\subseteq N(y)$, then one can remove $x$ from $G$ without changing the size of an optimal solution of \textsc{Uniquely Restricted Matching}. Additionally, we use the following result, which follows from Sperner’s Theorem and Stirling’s approximation: given a set $X$ of cardinality $N$, and a collection $Y$ of subsets of $X$ such that any two subsets in $Y$ are incomparable\footnote{Two sets $A$ and $B$ are \emph{incomparable} if neither $A \subseteq B$ nor $B \subseteq A$.}, it holds that $|Y| \leq \frac{2^N}{\sqrt{N}}$.

Applying this to our setting, since any two vertices in $I$ are incomparable (with respect to their neighborhoods in $X$), we obtain that $|I| \leq \frac{2^t}{\sqrt{t}}$, where $t = |X|$. Consequently, the total number of vertices in $G$ satisfies $|V(G)| \leq t + \frac{2^t}{\sqrt{t}}$. Thus, \textsc{Uniquely Restricted Matching} admits an exponential kernel when parameterized by the vertex cover number. However, we show that obtaining a polynomial kernel is unlikely unless $\mathsf{NP} \subseteq \mathsf{coNP} \slash \mathsf{poly}$, even by considering the size of the matching as an additional parameter. Formally, we prove the following.

 \begin{restatable}{theorem} {vcthm}
    \textsc{Uniquely Restricted Matching} does not admit a polynomial kernel when parameterized by the vertex cover number plus the size of the matching unless $\mathsf{NP}\subseteq \mathsf{coNP} \slash \mathsf{poly}$.
\end{restatable}
\subsection{Future Directions}
A natural question to explore is whether it is possible to design an algorithm for \textsc{Uniquely Restricted Matching} parameterized by treewidth that runs in time $2^{o(\tw^2)} \cdot |V(G)|^{\mathcal{O}(1)}$. Additionally, an intriguing direction is to investigate the potential of the \textit{clique-width} parameter, which has recently been studied in the context of related problems like \textsc{Induced Matching} and \textsc{Acyclic Matching} \cite{lampis2025structural}. It will be interesting to identify whether the same techniques can be extended to \textsc{Uniquely Restricted Matching}. Another promising area of research involves studying guarantee parameters of the form $\textsf{UB}-\ell$, where $\textsf{UB}$ is an upper bound on the uniquely restricted matching number. This approach has been explored in the context of \textsc{Induced Matching} and \textsc{Acyclic Matching} \cite{DBLP:conf/stacs/Koana23,chaudhary2025parameterized}, and it appears particularly relevant given that the problem is W[1]-hard with respect to the solution size $\ell$. 

\subsection{Organization of the Paper}
In Section~\ref{sec:prelim}, we present the definitions and concepts that will be used throughout the paper. Section~\ref{sec:line} focuses on the \textsc{Uniquely Restricted Matching} problem on line graphs, with solution size as the parameter. In Section~\ref{sec:treewidth}, we analyze the parameterized complexity of the problem with respect to the treewidth of the input graph. Section~\ref{sec:vc} investigates the problem under the vertex cover number as the parameter. \autoref{tablepig2} provides a summary of these results.

\begin{table}[h]
    \centering
    \scalebox{0.85}{
    \begin{tabular}{|l|c|c|}
        \hline
        & \underline{\textbf{Parameter}} &\underline{\textbf{Result}} \\
        
        1. & Solution Size & \textsf{FPT} on line graphs  \\
        
        2. & Treewidth  & \textsf{FPT}  \\
        
        3. & Vertex Cover Number plus the Solution Size & Does not admit a polynomial kernel unless $\mathsf{NP} \subseteq \mathsf{coNP} \slash \mathsf{poly}$ \\
        \hline
    \end{tabular}
    }
    \caption{Overview of parameterized results for \textsc{Uniquely Restricted Matching} presented in this paper. }
    \label{tablepig2}
\end{table}

\section{Preliminaries}\label{sec:prelim}
\subsection{Graph-theoretic Notations and Definitions}\label{GT}
All graphs considered in this paper are simple, undirected, and connected unless stated otherwise. 	Standard graph-theoretic terms not explicitly defined here can be found in \cite{diestel2012graph}. For a graph $G$, let $V(G)$ denote its vertex set, and $E(G)$ denote its edge set. Given a graph $G$ and a subset of edges $E' \subseteq E(G)$, the \emph{edge-induced subgraph} $G[E']$ is the subgraph of $G$ with edge set $E'$ and vertex set consisting of all endpoints of edges in $E'$. Given a matching $M$, a vertex
$v\in V(G)$ is \textit{$M$-saturated} if $v$ is incident on an edge of $M$, that is, $v$ is an end
vertex of some edge of $M$. Given a graph $G$ and a matching $M$, we use
the notation $V_{M}$ to denote the set of $M$-saturated vertices and $G[V_{M}]$ to denote the
subgraph induced by $V_{M}$. A matching that saturates all the vertices of a graph is called a \textit{perfect matching}. If $uv\in M$, then
$v$ is called the \textit{$M$-mate} of $u$ and vice versa. 

Given a graph $G$, an \emph{even cycle} (i.e., a cycle with an even number of edges) in $G$ is said to be an \emph{alternating cycle} with respect to a matching $M$ if every second edge of the cycle belongs to $M$. The following proposition characterizes uniquely restricted matchings in terms of alternating cycles.

\begin{proposition} [\cite{golumbic2001uniquely}]\label{defurm}
Let $G$ be a graph. A matching $M$ in $G$ is uniquely restricted if and only if there is no alternating cycle with respect to $M$ in $G$.
\end{proposition}

The \emph{line graph} of a graph $H$, denoted by $L(H)$, is constructed by representing each edge of $H$ as a vertex in $L(H)$, and two vertices in $L(H)$ are adjacent if and only if their corresponding edges in $H$ share a common endpoint. A graph $G$ is said to be a \emph{line graph} if there exists a graph $H$ such that $G=L(H)$. 
 A \emph{vertex cover} of a graph $G$ is a subset $S\subseteq V(G)$ such that every edge of $G$ has at least one of its endpoints in $S$, and the size of a smallest vertex cover among all vertex covers of $G$ is known as the \emph{vertex cover number} of $G$. 
 
 Let $G$ be a graph. A \emph{coloring} on a set $X\subseteq V(G)$ is a function $f:X\rightarrow S$, where $S$ is any set. Here, the elements of $S$ are called \emph{colors}. For a coloring $f$ on $X\subseteq V(G)$ and $Y\subseteq X$, we use the notation $f|_{Y}$ to denote the restriction of $f$ to $Y$. For a coloring $f:X(\subseteq V(G)) \rightarrow S$, a vertex $v\in V(G)\setminus X$, and a color
$\alpha\in S$, we define $f_{v\rightarrow \alpha}:X \cup \{v\}\rightarrow S$ as follows:
$$
  f_{v\rightarrow \alpha}(x) = \begin{cases}
        f(x) \ \ \ \ \text{\ if \ $x\in X$,}
        \\
        \alpha \ \ \  \ \ \ \ \    \text{\ if \ $x=v$.}
           
        \end{cases}
$$
More generally, for a coloring $f:X(\subseteq V(G))\rightarrow \{0,1,2\}$, a set $Y\subseteq X$, and a color
$\alpha\in \{0,1,2\}$, we define $f_{Y\rightarrow \alpha}:X \rightarrow \{0,1,2\}$ as follows:

$$
  f_{Y\rightarrow \alpha}(x) = \begin{cases}
        f(x) \ \ \ \ \text{\ if \ $x\notin Y$,}
        \\
        \alpha \ \ \  \ \ \ \ \    \text{\ if \ $x\in Y$.}
           
        \end{cases}
$$

\subsection{Parameterized Complexity}

Standard notions in Parameterized Complexity not explicitly defined here can be found in \cite{cygan2015parameterized}. In the framework of Parameterized Complexity, each instance of a problem $\mathrm{\Pi}$ is associated with a non-negative integer \textit{parameter} $k$. A parameterized problem $\mathrm{\Pi}$ is \textit{fixed-parameter tractable} ($\mathsf{FPT}$) if there is an algorithm that, given an instance $(I,k)$ of $\mathrm{\Pi}$, solves it in time $f(k)\cdot |I|^{\mathcal{O}(1)}$, for some computable function $f(\cdot)$. A parameterized (decision) problem $\mathrm{\Pi}$ admits a \emph{kernel} of size $f(k)$ for some function $f$ that depends only on $k$ if the following is true: There exists an algorithm (called a \emph{kernelization algorithm}) that runs in $(|I|+k)^{\mathcal{O}(1)}$ time and translates any input instance $(I,k)$ of $\mathrm{\Pi}$ into an equivalent instance $(I',k')$ of $\mathrm{\Pi}$ such that the size of $(I',k')$ is bounded by some computable function $f(k)$. If the function $f$ is polynomial in $k$, then the problem is said to admit a \emph{polynomial kernel}. It is well-known that a decidable parameterized problem is \FPT if and only if it admits a kernel \cite{cygan2015parameterized}. To show that a parameterized problem does not admit a polynomial kernel up to reasonable complexity assumptions can be done by making use of the cross-composition technique introduced by Bodlaender, Jansen, and Kratsch \cite{bodlaender2010cross}. To apply the cross-composition technique, we need the following additional definitions. Since we are using an OR-cross-composition, we give definitions tailored for this case.

\begin{definition}[Polynomial Equivalence Relation, \cite{bodlaender2010cross}]\label{def:PolyEquivRel}
An equivalence relation $\mathcal{R}$ on $\mathrm{\Sigma^{*}}$ is a
\emph{polynomial equivalence relation} if the following two conditions hold:
\begin{enumerate}
    \item There is an algorithm that given two strings $x,y \in \mathrm{\Sigma^{*}}$ decides whether $x$ and $y$ belong to the same equivalence class in $(|x|+|y|)^{\mathcal{O}(1)}$ time.
\item For any finite set $S\subseteq \mathrm{\Sigma^{*}}$ the equivalence relation $\mathcal{R}$ partitions the elements of $S$ into at
most $(\max_{x\in S} |x|)^{\mathcal{O}(1)}$ classes.
\end{enumerate}
\end{definition}

\begin{definition} [OR-Cross-composition, \cite{bodlaender2010cross}] 
Let $L\subseteq \mathrm{\Sigma^{*}}$ be a problem and let $Q\subseteq \mathrm{\Sigma^{*}} \times \mathbb{N}$ be a parameterized problem. We say that $L$ \emph{OR-cross-composes} into $Q$ if there is a polynomial equivalence
relation $\mathcal{R}$ and an algorithm which, given $t$ strings $x_1, x_2, \ldots, x_t$ belonging to the same equivalence class of $\mathcal{R}$, computes an instance $(x^{*},k^{*})\in \mathrm{\Sigma^{*}}\times \mathbb{N}$ in time polynomial in $\Sigma_{i=1}^{t}|x_{i}|$
such that:
\begin{enumerate}
\item $(x^{*},k^{*})\in Q \iff$ $x_{i}\in L$ for some $i\in [t]$,
\item $k^{*}$ is bounded by a polynomial in $\max_{i=1}^{t} |x_{i}|+\log t$.
\end{enumerate}
\end{definition}

\begin{proposition} [\cite{bodlaender2010cross}] \label{orcross}
If some problem $L$ is $\mathsf{NP}$-$\mathsf{hard}$ under Karp reduction and there exists an $\emph{OR-cross-composition}$ from $L$ into
some parameterized problem $Q$, then there is no polynomial kernel for $Q$ unless $\mathsf{NP}\subseteq \mathsf{coNP} \slash \mathsf{poly}$.
\end{proposition}
\noindent \textbf{Treewidth.} A \emph{rooted tree} $T$ is a tree having a distinguished vertex labeled $r$, called the \emph{root}. For a vertex, $v\in V(T)$, an \emph{$(r, v)$-path} in $T$ is a sequence of distinct vertices starting from $r$ and ending at $v$ such that every two consecutive vertices are connected by an edge in the tree. The \emph{parent} of a vertex $v$ different from $r$ is its
neighbor on the unique $(r, v)$-path in $T$. The other neighbors of $v$ are its \emph{children}. A vertex $u$ is an \emph{ancestor} of $v$ if $u\neq v$ and $u$ belongs on
the unique $(r, v)$-path in $T$. A \emph{descendant} of $v$ is any vertex $u\neq v$ such that $v$ is its ancestor. The subtree rooted at $v$ is the subgraph of $T$ induced by $v$ and its descendants.

\begin{definition} [Tree Decomposition] 
A \emph{tree decomposition} of a graph $G$ is a pair $\mathcal{T}=(\mathbb{T},\{\mathcal{B}_{x}\}_{x\in V(\mathbb{T})})$, where $\mathbb{T}$ is a rooted tree and each $\mathcal{B}_{x},x\in V(\mathbb{T})$, is a subset of $V(G)$ called a \emph{bag}, such that 
\begin{enumerate}
\item [$T.1)$] $\bigcup_{x\in V(\mathbb{T})}\mathcal{B}_{x}=V(G)$,
\item [$T.2)$] for any edge $uv\in E(G)$, there exists a node $x\in V(\mathbb{T})$ such that $u,v\in \mathcal{B}_{x}$, 
\item [$T.3)$] for all $x,y,z\in V(\mathbb{T})$, if $y$ is on the path from $x$ to $z$ in $\mathbb{T}$ then $\mathcal{B}_{x}\cap \mathcal{B}_{z}\subseteq \mathcal{B}_{y}$.
\end{enumerate}
\end{definition}

The \emph{width} of a tree decomposition $\mathcal{T}$ is the size of its largest bag minus one. The \emph{treewidth} of $G$ is the minimum width over all tree decompositions of $G$. We denote the treewidth of a graph by $\tw$.

Dynamic programming algorithms on tree decompositions are often presented on \emph{nice tree decompositions}, which
were introduced by Kloks \cite{kloks}. We refer to the tree decomposition definition given by Kloks as a standard nice tree
decomposition, which is defined as follows: 

\begin{definition} [Nice Tree Decomposition] 
	Given a graph $G$, a tree decomposition 
$\mathcal{T} = (\mathbb{T},\{\mathcal{B}_{x}\}_{x\in V(\mathbb{T})})$ of $G$ is a \emph{nice tree decomposition} if the following hold:
\begin{enumerate}
	\item [$N.1)$] $\mathcal{B}_{r}=\emptyset$, where $r$ is the root of $\mathbb{T}$,
	and $\mathcal{B}_{l}=\emptyset$ for every leaf $l$ of $\mathbb{T}$.
	\item [$N.2)$] Every non-leaf node $x$ of $\mathbb{T}$ is of one of the following types:
	\begin{enumerate}
		\item [$N.2.1)$]
		\noindent \textbf{Introduce vertex node:} $x$ has exactly one child $y$, and $\mathcal{B}_{x}=\mathcal{B}_{y}\cup \{v\}$ where $v\notin \mathcal{B}_{y}$. We say that $v$ is \emph{introduced} at $x$.
		
		\item [$N.2.2)$] \noindent \textbf{Forget vertex node:} $x$ has exactly one child $y$, and $\mathcal{B}_{x}=\mathcal{B}_{y}\setminus \{u\}$ where $u\in \mathcal{B}_{y}$. We say that $u$ is \emph{forgotten} at $x$.
		
		\item [$N.2.3)$] \noindent \textbf{Join node:}  $x$ has exactly two children, $y_{1}$ and $y_{2}$, and $\mathcal{B}_{x}=\mathcal{B}_{y_{1}}=\mathcal{B}_{y_{2}}$.
	\end{enumerate}	
	\end{enumerate}
\end{definition}

\begin{observation} \label{joinobs}
All the common vertices in the bags of the subtrees of the children of a join node appear in the bag of the join node.
\end{observation}

For our problems, we want the standard nice tree decomposition to satisfy an additional property, and that is, among the vertices present in the bag of a join node, no edges have been introduced yet.  To achieve this, we use another known
type of node, an \emph{introduce edge node}, which is defined as follows:\\

	 \noindent \textbf{Introduce edge node:} $x$ has exactly one child $y$, and $x$ is labeled with an edge $uv\in E(G)$ such that $u,v\in \mathcal{B}_{x}$ and $\mathcal{B}_{x}=\mathcal{B}_{y}$. We say that $uv$ is \emph{introduced} at $x$.

The use of introduce edge nodes enables us to add edges one by one in our nice tree decomposition. We also need our nice tree decomposition to satisfy the following additional property. 

\begin{definition}[Deferred Edge Property]
A nice tree decomposition that includes introduce edge nodes is said to exhibit the \emph{deferred edge} property if for every edge \( uv \in E(G) \), the corresponding introduce edge bag is placed as late as possible in the decomposition.  Specifically, if \( t(v) \) and \( t(u) \) are the highest nodes where \( v \) and \( u \) respectively appear $($with \( t(v) \) being an ancestor of \( t(u) \)$)$, then the introduce edge bag for \( uv \) is positioned at any node between \( t(u) \) and its parent $($which forgets \( u \)$)$, ensuring that edges are introduced at the latest possible stage.
\end{definition}

\noindent\textbf{Note:} In the above definition, suppose vertex $u$ has incident edges $uv$ and $uw$, with $t(u)$ appearing before both $t(v)$ and $t(w)$. When adding the edge node for $uv$, it must be placed between $t(u)$ and its parent, say $p(t(u))$. Later, the edge node for $uw$ can be inserted either between $t(u)$ and the $uv$ node, or between the $uv$ node and $p(t(u))$. Moreover, we make sure that every edge is introduced exactly once. 

Now, consider the following observation.

\begin{observation}[\cite{cygan2015parameterized,kloks}]\label{obs:deferred}
A standard nice tree decomposition of a graph \( G \) can be transformed into a variant satisfying the \emph{deferred edge} property in \( \tw^{\mathcal{O}(1)} \cdot |V(G)| \) time, without asymptotically increasing the number of nodes.

\end{observation}

For each node $x$ of the tree decomposition, let $V_{x}$ be the union of all the bags present in the subtree of $\mathbb{T}$ rooted at $x$, including $\mathcal{B}_{x}$.
For each node $x$ of the tree decomposition, define the subgraph $G_{x}$ of $G$ as follows:

$G_{x}=(V_{x},E_{x}=\{e: e$ is introduced in the subtree of $\mathbb{T}$ rooted at $x\})$.

\section{Parameter: Solution Size for Line Graphs} \label{sec:line}
In this section, we prove that \textsc{Uniquely Restricted Matching} admits a fixed-parameter tractable (FPT) algorithm on line graphs when parameterized by the \emph{solution size}. In order to prove that, we need the following lemma.
\liner*
 \begin{proof}
We will prove the necessity part as the sufficiency part is easy to see by the definition of line graphs and urm. Let $M=\{e_{1},\ldots, e_{\ell}\}$  be a urm in $G$ size $\ell$. Observe that each matched edge in $G$ corresponds to a path of length 2 (i.e., a $P_3$) in $H$, where the two endpoints of the matched edge in $G$ represent two adjacent edges of this $P_3$ in $H$. The corresponding $P_3$'s in $H$, derived from the edges in $M$, must be edge-disjoint; otherwise, $M$ would not be a matching in $G$.

Now consider a pair of edge-disjoint $P_3$'s in $H$, say $W_i$ and $W_j$, that together form an edge-induced subgraph isomorphic to $K_{1,4}$. In the line graph $G$, the four edges involved in $W_i \cup W_j$ correspond to a clique on four vertices. This implies that the two matched edges $e_i$ and $e_j$ (from $M$) lie within a 4-clique in $G$, which is not possible as it would induce an alternating cycle, contradicting the definition of a urm due to Proposition \ref{defurm}.

Finally, we show that $\bigcup_{i \in [\ell]} W_i$ forms a forest in $H$. Suppose, for contradiction, that $\bigcup_{i\in [\ell]}{W_i}$ contains a cycle, say $C$ in $H$. Then, for any $W_i$ participating in $C$, the two edges of it can interact with the cycle in at most two ways: either both edges of $W_i$ lie on the cycle, or exactly one edge lies on the cycle.

We proceed by a case analysis to derive a contradiction to the assumption that $M$ is a urm in $G$.

First, consider the case where every $W_i$ participating in $C$ has both its edges included in the cycle. Take the smallest case that involves exactly two such $W_i$'s, say $W_i$ and $W_j$ in $C$. In this scenario, an alternating cycle will be formed in $G$. See Fig.~\ref{fig1}(a) for the illustration. Observe that the same reasoning extends to cases where more than two $P_3$'s are involved in 
$C$, with each $P_3$ contributing both of its edges to the cycle.

\begin{figure}[ht]
 \centering
    \includegraphics[scale=0.7]{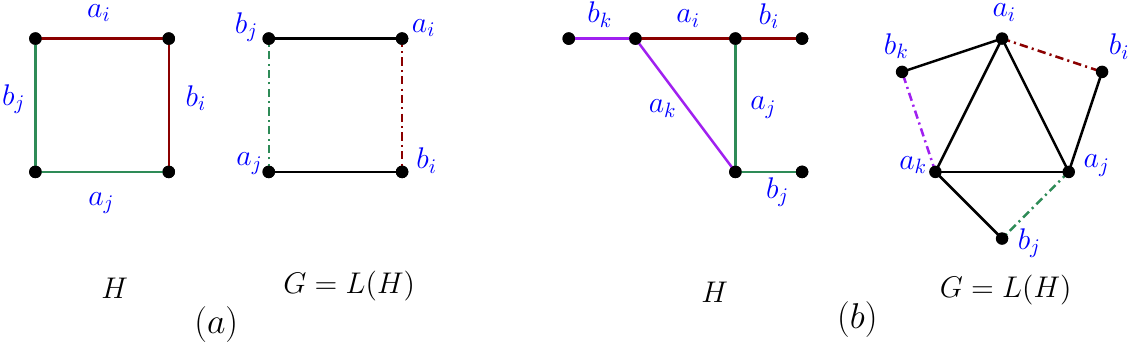}
    \caption{The dotted edges represent matching edges. Here, $W_{i}=a_{i}b_{i}$, $W_{j}=a_{j}b_{j}$, and $W_{k}=a_{k}b_{k}$. Note that the edges of $H$ correspond to vertices in $G$.}
    \label{fig1}
\end{figure}

Next, consider the case where each participating $W_i$ contributes exactly one edge to the cycle. Again, we focus on the smallest instance involving exactly three $P_3$'s. Here again, an alternating cycle will be formed in $G$. See Fig.~\ref{fig1}(b) for the illustration. Observe that the same reasoning extends to cases where more than three $P_3$'s are involved in 
$C$, with each $P_3$ contributing exactly one of its edges to the cycle. Also, observe that the restriction preventing any two distinct paths $W_i$ and $W_j$ (for $i \neq j$) from together inducing a $K_{1,4}$ ensures that the resulting cycle has a very specific and constrained structure.

Next, consider the more general case where the $P_3$'s participating in $C$ contribute in a mixed manner, with some contributing both edges and others exactly one edge. Note that the specific arrangement of paths does not affect the outcome. This is because of the properties of line graphs: if a path $W_i$ (a $P_3$) contributes only one edge to the cycle $C$, then the vertex corresponding to its other edge (which is not part of $C$) is always adjacent to a vertex from another $P_3$ in $C$ that shares a common edge with the participating edge of $W_i$. For example, in Fig.~\ref{fig6}, in $G$, the vertex $b_i$ is adjacent to $a_j$. In $H$, the edge $a_j$ is part of another $P_3$ that participates in $C$ and is adjacent to $a_i$. Thus, there is an edge in $G$ between the vertices corresponding to $b_i$ and $a_j$. Consequently, it is always possible to take a detour and construct an alternating cycle in $G$ using the edges that appear in $H$ due to the presence of such adjacency. Therefore, we conclude that $\bigcup_{i \in [\ell]} W_i$ must form a forest in $G$.
\end{proof}
\begin{figure}[t]
 \centering
    \includegraphics[scale=0.7]{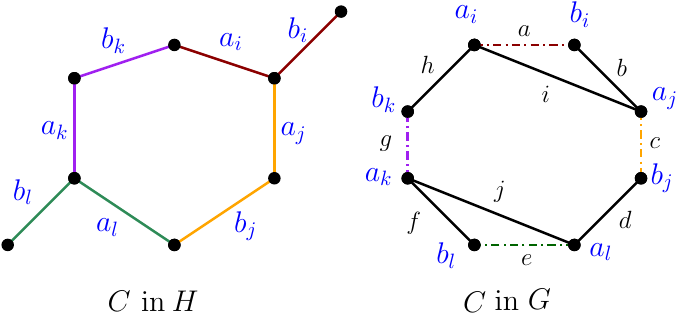}
    \caption{The cycle $C=a_i,a_j,b_j,a_l,a_k,b_k,a_i$ in $H$ corresponds to the cycle $i,c,d,j,g,h,i$ in $G$. The dotted edges represent matching edges. The detour in $G$, given by $a,b,c,d,e,f,g,h,a$, is made possible due to the presence of edges $b$ and $f$.}
    \label{fig6}
\end{figure}

% Next, we refer to a result by Alon et al.~\cite{alon1995color}, who applied the color-coding technique to demonstrate that the problem of finding a given forest $F$ with $k$ vertices as a subgraph (not necessarily an induced subgraph) in an arbitrary graph $G$ is fixed-parameter tractable with respect to the parameter $k$. This result can be stated as follows.

%  \begin{proposition}[\cite{alon1995color}]
%      Let $F$ be a forest on $k$ vertices and $G$ be a graph. A subgraph of $G$ isomorphic to $F$, if one exists, can be found in $2^{\mathcal{O}(k)} \cdot |V(G)| \log |V(G)|$ time.

%  \end{proposition}

Now, we recall the following definition: a \emph{partition} of a positive integer $n$ is a multiset of positive integers whose sum is $n$. The number of such partitions is denoted by $p(n)$. It is a well-known result that for every positive integer $n$, the number of partitions satisfies $p(n) \leq e^{3\sqrt{n}}$~\cite{de2009simple}.
 
 We are now ready to prove the main theorem of this section.

 \thmfpt*

 \begin{proof}
  Fix the line graph $G$ and integer $\ell$ as an instance of \textsc{Uniquely Restricted
Matching} problem. Let $H$ be a graph without isolated vertices such that $G = L(H)$. By \autoref{thm:liner}, $G$ has a urm of size $\ell$ if and only if $H$ contains $\ell$ edge-disjoint paths $W_1, \ldots, W_\ell$, each of length $2$ such that $\bigcup_{i\in [\ell]}{W_i}$ is a forest, and no two distinct paths $W_i$ and $W_j$ (for $i\neq j$) together induce an edge-induced subgraph isomorphic to $K_{1,4}$. Let $\mathcal{F}$ denote the set of all such forests.

Thus, the task is to generate all forests with $2\ell$ edges and, for each such forest, verify whether it satisfies the desired structure, all within FPT time. After that, it remains to check whether such a desired forest appears in $H$ as a subgraph or not. For this purpose, we refer to a result by Alon et al.~\cite{alon1995color}, who applied the color-coding technique to demonstrate that the problem of finding a given forest $F$ with $k$ vertices as a subgraph (not necessarily an induced subgraph) in an arbitrary graph $G$ is fixed-parameter tractable with respect to the parameter $k$. The result states that if $F$ is a forest on $k$ vertices and $G$ is a graph, then a subgraph of $G$ isomorphic to $F$, if one exists, can be found in $2^{\mathcal{O}(k)} \cdot |V(G)| \cdot \log |V(G)|$ time.

Note that the number of vertices in a forest with exactly $2\ell$ edges such that it has $\ell$ edge-disjoint $P_{3}$'s lies between $2\ell+1$ and $3\ell$.

Next, for each fixed $k$ satisfying $2\ell + 1 \leq k \leq 3\ell$, we consider all integer partitions of $k$. For each partition $(s_1, s_2, \ldots, s_t)$, we construct a forest $F = T_1 \cup \ldots \cup T_t$, where each $T_i$ is a tree with $s_i$ vertices. That is, $\sum_{i=1}^{t} \text{size}(T_i) = k$, and $\text{size}(T_i) = s_i$.

For a fixed tree size $s_i$, the number of possible unlabeled trees on $s_i$ vertices is at most $2^{\mathcal{O}(s_i)}$~\cite{otter1948number}. Therefore, for a single partition of $k$, the total number of distinct unlabeled forests we can construct is at most
$\prod_{i=1}^{t} 2^{\mathcal{O}(s_i)} = 2^{\mathcal{O}(\sum_{i=1}^{t} s_i)} = 2^{\mathcal{O}(k)}.$

Since the number of integer partitions of $k$ is at most $p(k)\leq e^{3\sqrt{k}}\leq 2^{\mathcal{O}(\sqrt{k})}$, the total number of candidate forests over all such $k$ is $|\mathcal{F}| \leq 2^{\mathcal{O}(\sqrt{k})} \cdot 2^{\mathcal{O}(k)} = 2^{\mathcal{O}(k)}.$ 

Now, for a given forest $F\in \mathcal{F}$, we can check in linear time whether it satisfies the required properties: that it has $\ell$ edge-disjoint paths $W_1, \ldots, W_\ell$, each of length $2$ such that $\bigcup_{i\in [\ell]}{W_i}$ is a forest, and no two distinct paths $W_i$ and $W_j$ (for $i\neq j$) together induce an edge-induced subgraph isomorphic to $K_{1,4}$ by the following greedy procedure. 
     
     We begin by considering a rooted tree from the forest (we will do it for every tree in the forest) and traversing it starting from an arbitrary leaf vertex of the tree. If the parent of this leaf (referred to as the co-leaf) has three or more pendant children, an edge-induced subgraph isomorphic to $K_{1,4}$ is formed, and we reject this forest. Otherwise, it follows that each co-leaf has either one or two pendant children. In either case, the co-leaf must serve as the middle vertex of a $P_3$. If the co-leaf has two children, these two edges form a $P_3$, which are then removed, and the process is repeated on the resulting reduced forest. If the co-leaf has one child, we attempt to identify a $P_3$ by traversing from the leaf to the co-leaf and then to the parent of the co-leaf. Once such a $P_3$ is found, we remove its edges and repeat the process; otherwise, we reject the tree.

 Let $\mathcal{F'} \subseteq \mathcal{F}$ denote the subset of forests that were accepted by the greedy algorithm described above. We now check for each forest in $\mathcal{F'}$ whether it appears as a subgraph in $H$ by using the color-coding technique of Alon et al~\cite{alon1995color}. If we find at least one such forest, we return YES, indicating that $G$ has a urm of size $\ell$ by Lemma~\ref{thm:liner}. Otherwise, if no forest in $\mathcal{F'}$ is found as a subgraph in $H$, we return NO, implying that $G$ does not contain a urm of size at least $\ell$. 

   Observe that since $k = \mathcal{O}(\ell)$ and every step can be performed within the claimed FPT time, this completes the proof of the theorem.
 \end{proof}
 % \section{Bipartite Graphs}
 % \todo[inline,color=yellow]{todo: Its complexity will solve the complexity for split graphs as well.}

\section{Parameter: Treewidth} \label{sec:treewidth}
%\subsection{FPT}

 In this section, we present an $\mathcal{O}(2^{\tw^2/2}\cdot |V(G)|)$ time algorithm for \textsc{Uniquely Restricted Matching}. Given an input graph $G$ and an integer $k$, it is known that in time $2^{\mathcal{O}(k)}\cdot |V(G)|$, we can either produce a tree decomposition of $G$ of width at most $2k+1$, or determine that the treewidth of $G$ exceeds $k$ \cite{korhonen2023single}. Therefore, by using a polynomial-time approximation algorithm, we can efficiently compute a tree decomposition, which can then be transformed into a nice tree decomposition
$\mathcal{T} = (\mathbb{T},\{\mathcal{B}_{x}\}_{x\in V(\mathbb{T})})$ of the graph $G$ and width $\tw$, satisfying the \textit{deferred edge} property. For this purpose, we introduce the following notations.

We use distinct labels, referred to as colors, to represent the possible states of a vertex in a bag $\mathcal{B}_{x}$ of $\mathcal{T}$ with respect to a matching $M$ as follows: 
\begin{itemize}
\item  \textbf{0}: A vertex colored 0 is not saturated by $M$.

\item  \textbf{1}: A vertex colored $1$ is saturated by $M$, and the edge between the vertex and its $M$-mate has also been introduced in $G_{x}$.

\item \textbf{2}: A vertex colored $2$ is saturated by $M$, and either its $M$-mate has not yet been introduced in $G_{x}$, or the edge between the vertex and its $M$-mate has not yet been introduced in $G_{x}$ (in the latter case, note that the $M$-mate also belongs to $\mathcal{B}_{x}$).
\end{itemize}

Now, we define the following notion.

\begin{definition}[Valid Coloring] \label{im1}
Given a node $x$ of $\mathbb{T}$, a coloring $f:\mathcal{B}_{x}\rightarrow\{0,1,2\}$ is \emph{valid} if there exists a coloring $\hat{f}:V_{x}\rightarrow\{0,1,2\}$ in $G_{x}$, called a \emph{valid extension} of $f$, such that the following hold: 
\begin{enumerate}
\item  $\hat{f}$ restricted to $\mathcal{B}_{x}$ is exactly $f$.
\item  Vertices colored $2$ under $\hat{f}$ must belong to $\mathcal{B}_{x}$.
\item Vertices colored 1 under $\hat{f}$ admit exactly one perfect matching. 

\end{enumerate}
\end{definition}

Note that since the vertices colored 1 under a valid coloring $f$ admit exactly one perfect matching under some valid extension of $f$, we can associate each such vertex with a unique $M$-mate, effectively forming virtual matched pairs. Now, we define a path to be an \emph{alternating path} if both its endpoints are either assigned colored 1 or both are colored 2, and there exists a path between them in $V_x$ such that all internal vertices are colored 1. Moreover, we require that the path alternates between matched and unmatched edges in $V_x$. See Fig.~\ref{fig2} for an illustration.

\begin{figure}[h]
 \centering
    \includegraphics[scale=0.7]{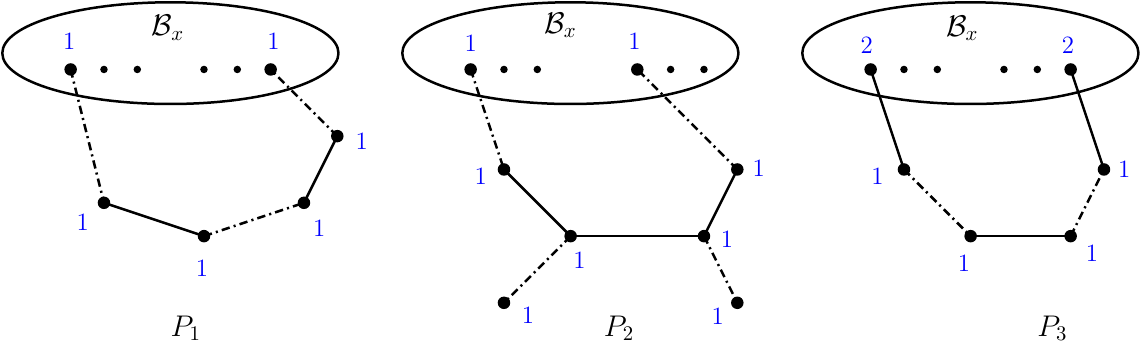}
    \caption{The dotted edges represent matching edges. Note that $P_1$ and $P_3$ are alternating paths, whereas $P_2$ is not.}
    \label{fig2}
\end{figure}

For a fixed bag $\mathcal{B}_{x}$, we define $P$ as a square matrix of size $|\mathcal{B}_{x}|\times |\mathcal{B}_{x}|$ with entries from $\{0,1\}$, where both the rows and columns correspond to elements of $\mathcal{B}_{x}$. We refer to such a matrix as an \emph{alternating matrix}. Intuitively, the alternating matrix $P$ encodes information about the existence of alternating paths between pairs of vertices in $\mathcal{B}_{x}$ within the graph $V_x$. To maintain consistency, we restrict our attention to symmetric matrices with diagonal elements fixed at zero.

 For a fixed bag $\mathcal{B}_{x}$, we say that a coloring function $f:\mathcal{B}_{x}\rightarrow\{0,1,2\}$ and an alternating matrix $P:\mathcal{B}_{x}\times \mathcal{B}_{x}\rightarrow \{0,1\}$ are \emph{compatible} if:
\begin{itemize}
  
    \item If $f(u)=0$ for a vertex $u\in \mathcal{B}_{x}$, then all entries in $P$ corresponding to $u$ must be set to $0$.
    \item An entry $P(u,v)=1$ indicates the existence of an alternating path between vertices $u$ and $v$ in the graph $G_x$ whereas $P(u,v)=0$ signifies that no such alternating path exists. 
\end{itemize}

We have a table $\mathcal{M}$ with an entry $\mathcal{M}_{x}[f,P]$ for each node $x$ of $\mathbb{T}$, for every coloring $f:\mathcal{B}_{x}\rightarrow\{0,1,2\}$, and for every alternating matrix $P$. The number of possible choices for $x$ is at most $\mathcal{O}(\tw\cdot n)$, the number of possible colorings $f$ is at most $3^{\tw}$, and since we consider symmetric matrices with 0's on the diagonal, and the number of possible matrices $P$ is at most $\mathcal{O}(2^{\frac{\mathsf{tw}^{2}}{2}})$. Consequently, the total size of the table $\mathcal{M}$ is bounded by $\mathcal{O}(2^{\mathsf{tw}^{2}}\cdot 
|V(G)|)$. The following definition specifies the value stored in each entry $\mathcal{M}_{x}[f,P]$ of $\mathcal{M}$.
\begin{definition}
If $f$ is valid, the entry $\mathcal{M}_{x}[f,P]$ stores the maximum number of vertices colored $1$ or $2$ under some valid extension $\hat{f}$ of $f$ in $G_{x}$, ensuring that $f$ remains compatible with $P$. Otherwise, the entry $\mathcal{M}_{x}[f,P]$ stores the value $-\infty$ and marks $f$ as \emph{invalid}.
\end{definition}

Since the root of $\mathbb{T}$ is an empty node, note that the maximum number of vertices saturated by any uniquely restricted matching is exactly $\mathcal{M}_{r}[\emptyset,\emptyset]$, where $r$ is the root of the decomposition $\mathcal{T}$. 

Next, we present recursive formulas for computing the entries of the table $\mathcal{M}$. 
\bigskip

\noindent \textbf{Leaf node:} For a leaf node $x$, we have that $\mathcal{B}_{x}=\emptyset$. Hence, there is only one possible coloring on $\mathcal{B}_{x}$, that is, the empty coloring\footnote{A coloring defined on an empty set is an \emph{empty coloring.}}. Also, since $\mathcal{B}_{x}=\emptyset$, the only possible alternating matrix is $P=\emptyset$.
Thus, we have $\mathcal{M}_{x}[\emptyset,\emptyset]=0$. 

\medskip
    
\noindent \textbf{Introduce vertex node:} Suppose that $x$ is an introduce vertex node with child node $y$ such that $\mathcal{B}_{x}=\mathcal{B}_{y}\cup \{v\}$ for some $v\notin \mathcal{B}_{y}$. Note that we have not introduced any edges incident on $v$ so far, so $v$ is isolated in $G_x$. Furthermore, for any subset $Y\subseteq X$, define the restriction of $P$ to $Y$, denoted as $P|_Y$, as the matrix obtained by removing the rows and columns of $P$ that are not in $Y$. Now, for every coloring $f: \mathcal{B}_{x}\rightarrow \{0,1,2\}$ and for every alternating matrix $P:\mathcal{B}_{x}\times \mathcal{B}_{x} \rightarrow \{0,1\}$, we have the following recursive formula.

\begin{equation*}
  \mathcal{M}_{x}[f,P] = \begin{cases}
        \mathcal{M}_{y}[f|_{\mathcal{B}_{y}},P|_{\mathcal{B}_{y}}] \ \  \ \ \ \ \ \ \text{\ if \ $f(v)=0$,}
        \\
        -\infty \ \ \  \ \ \ \ \ \ \ \ \ \ \ \ \ \ \ \ \ \ \ \  \text{\ if \ $f(v)=1$,}
            \\
            \mathcal{M}_{y}[f|_{\mathcal{B}_{y}},P|_{\mathcal{B}_{y}}]+1 \ \ \ \text{\ if \ $f(v)=2$}.
        \end{cases}
 \end{equation*}

\bigskip

Note that when $f(v)=1$, then $f$ is invalid as $v$ does not have any neighbor in $G_{x}$ (by the definition of a valid coloring, $v$ needs one neighbor of color $1$ in $G_{x}$), and hence $\mathcal{M}_{x}[f,P]=-\infty$. 

Next, when $f(v)=0$ or $f(v)=2$, then $f$ is valid if and only if $f|_{\mathcal{B}_{y}}$ is valid. Moreover, when $f(v)=2$, we increment the value by one as one more vertex has been colored $2$ in $G_{x}$.

Clearly, the evaluation of all introduce vertex nodes can be done in $\mathcal{O}(2^{\tw^2/2} \cdot |V(G)|)$ time.
\bigskip

\noindent \textbf{Introduce edge node:} 
Suppose that $x$ is an introduce edge node that introduces an edge $uv$, and let $y$ be the child of $x$.  For every coloring $f: \mathcal{B}_{x}\rightarrow \{0,1,2\}$ and for every alternating matrix $P:\mathcal{B}_{x}\times \mathcal{B}_{x} \rightarrow \{0,1\}$, we consider the following cases:

If $f(u)=f(v)=1$ and the entry corresponding  to $\{u,v\}$ in $P$ is 0, then
\begin{equation*}
	\mathcal{M}_{x}[f,P]=
\max \{\mathcal{M}_{y}[f_{\{u,v\}\rightarrow 2},P], \mathcal{M}_{y}[f,P]\}.
\end{equation*}

Else, if $f(u)=f(v)=1$ and the entry corresponding  to $\{u,v\}$ in $P$ is 1, then

 \begin{equation*}\mathcal{M}_{x}[f,P]=-\infty.\end{equation*}

 Else, 
	\begin{equation*}
	\mathcal{M}_{x}[f,P]=
\mathcal{M}_{y}[f,P]. 
\end{equation*}

If either $f(u)$ or $f(v)$ is $0$, then $f$ is valid if and only if $f$ is valid on $\mathcal{B}_{y}$. Next, let us consider the case when both $f(u)$ and $f(v)$ are $1$ and $f$ is valid. In this case, both $u$ and $v$ must be colored $2$ under $f$ in $\mathcal{B}_{y}$ (this follows by the definition of a valid coloring).

Clearly, the evaluation of all introduce edge nodes can be done in $\mathcal{O}(2^{\tw^2/2} \cdot |V(G)|)$ time.
  
\bigskip

\noindent  \textbf{Forget node:} Suppose that $x$ is a forget vertex node with
a child $y$ such that $\mathcal{B}_{x}=\mathcal{B}_{y}\setminus \{u\}$ for some $u\in \mathcal{B}_{y}$. For every coloring $f: \mathcal{B}_{x}\rightarrow \{0,1,2\}$ and for every alternating matrix $P:\mathcal{B}_{x}\times \mathcal{B}_{x} \rightarrow \{0,1\}$, we have
 
\begin{equation} \label{forgetim}\mathcal{M}_{x}[f,P]=
\displaystyle\max\{\mathcal{M}_{y}[f_{u \rightarrow 0},P],\mathcal{M}_{y}[f_{u \rightarrow 1},P]\} .
\end{equation}

 The first term on the right-hand side in (\ref{forgetim}) corresponds to the case when $f(u)=0$ in $\mathcal{B}_{y}$, and the second term corresponds to the case when $f(u)=1$ in $\mathcal{B}_{y}$. Note that the maximum is taken over colorings $f_{u \rightarrow 0}$  and $f_{u \rightarrow 1}$  only, as the coloring $f_{u \rightarrow 2}$ cannot be extended to a valid coloring once $u$ is forgotten. 

Clearly, the evaluation of all forget nodes can be done in $\mathcal{O}(2^{\tw^2/2} \cdot |V(G)|)$ time.

\bigskip
To define the join nodes, consider the following definition.
 \begin{definition} \label{prelimdef}
Given a graph $G$ and a set $X\subseteq V(G)$, two colorings $f_{1},f_{2}: X \rightarrow \{0,1,2\}$ are \emph{$f$-correct} for a coloring $f: X \rightarrow \{0,1,2\}$ if the following conditions hold:
\begin{enumerate}
	\item  $f(v)=0$ if and only if $f_{1}(v)=f_{2}(v)=0$,
	\item  $f(v)=1$ if and only if  $(f_{1}(v),f_{2}(v))\in\{(1,2),(2,1)\}$, and
	\item  $f(v)=2$ if and only if $f_{1}(v)=f_{2}(v)=2$.
\end{enumerate}
\end{definition}

\begin{definition}
    Let $x$ be a join node with children $y_{1}$ and $y_{2}$. For every coloring $f,f_{1},f_{2}: \mathcal{B}_{x}\rightarrow \{0,1,2\}$ and for every alternating matrices $P,P_1,P_2:\mathcal{B}_{x}\times \mathcal{B}_{x} \rightarrow \{0,1\}$, if $f_{1},f_{2}$ are $f$-correct, then $P_{1},P_{2}$ are $P$-\emph{right} for an alternating matrix $P$ if \begin{itemize}
        \item $P$ is compatible with $f$ and $P_{i}$ is compatible with $f_{i}$ for every $i\in [2]$.
        \item If $\{u,v\}$ has their corresponding value 1 in either $P_{1}$ or $P_{2}$, then their corresponding value in $P$ is also 1. Here, $u,v\in \mathcal{B}_{x}$ and $f(u)=f(v)=1$.    \end{itemize}
    
    \end{definition}
\bigskip

\noindent \textbf{Join node:}
Let $x$ be a join node with children $y_{1}$ and $y_{2}$. For every coloring $f: \mathcal{B}_{x}\rightarrow \{0,1,2\}$ and for every alternating matrix $P:\mathcal{B}_{x}\times \mathcal{B}_{x} \rightarrow \{0,1\}$, we have 
\begin{equation*}
\mathcal{M}_{x}[f,P]=\displaystyle\max_{f_{1},f_{2}} \{\mathcal{M}_{y_{1}}[f_{1},P_1]+\mathcal{M}_{y_{2}}[f_{2},P_2]-|f^{-1}(1)|-|f^{-1}(2)| \},
\end{equation*}

where $f_{1}: \mathcal{B}_{y_{1}}\rightarrow \{0,1,2\}$ and $f_{2}: \mathcal{B}_{y_{2}}\rightarrow \{0,1,2\}$ such that $f_{1}$ and $f_{2}$ are $f$-correct, and $f_{i}$ is compatible with $P_{i}$ for every $i\in [2]$. Also, $P_{1}$ and $P_{2}$ are $P$-right.

\begin{remark}
Note that incompatible pairs may arise.
\end{remark}

Note that we can restrict our search space to colorings $f_1$ and $f_2$ that are $f$-correct, as each edge is introduced exactly once and specifically at the latest possible stage. In other words, in a join node, the vertices assigned color 1 form an independent set. To compute the value of $\mathcal{M}_{x}[f,P]$, we look up the corresponding colorings in the child nodes $y_{1}$ and $y_{2}$, add their respective values, and subtract the number of vertices colored $1$ or $2$ under $f$. This subtraction is crucial; without it, the vertices colored $1$ or $2$ in $\mathcal{B}_{x}$ would be double-counted due to Observation \ref{joinobs}.

By the naive method, the evaluation of all join nodes altogether can be done in $\mathcal{O}(2^{\tw^2/2} \cdot |V(G)|)$ time.

 Thus, from the description of all nodes, we have the following theorem.
\treewidththm*

%\subsection{ETH Lower Bound}

\section{Parameter: Vertex Cover Number} \label{sec:vc}

In this section, we prove that under the assumption that $\mathsf{NP} \nsubseteq \mathsf{coNP} \slash \mathsf{poly}$, there does not exist
a polynomial kernel for \textsc{Uniquely Restricted Matching} when parameterized by the vertex cover number of the input graph plus the size of the matching. For this purpose, we give an OR-cross-composition (see Preliminaries) from the \textsc{Exact-3-Cover} problem (defined below), which is $\mathsf{NP}$-$\mathsf{hard}$ \cite{garey}.\\

\noindent\fbox{ \parbox{160mm}{
		\noindent\underline{\textsc{Exact-3-Cover}}:
		
		\smallskip
		\noindent\textbf{Instance:} A set $\mathcal{U}$ with $|\mathcal{U}|=3c$, where $c\in \mathbb{N}$, and a collection $\mathcal{X}$ of 3-element subsets of $\mathcal{U}$.
		
		\noindent\textbf{Question:} Does there exist a subcollection $\mathcal{X'}\subseteq\mathcal{X}$ such that every element of $\mathcal{U}$ appears in exactly one member of $\mathcal{X'}$?}}
\medskip

 We remark that our reduction is inspired by the reduction given by Chaudhary and Zehavi \cite{chaudhary2025parameterized} to prove that
there is no polynomial kernel for \textsc{Acyclic Matching} when parameterized by the vertex cover number plus the size of the matching of the input graph unless $\mathsf{NP}\subseteq \mathsf{coNP} \slash \mathsf{poly}$. 

Now, consider the following construction.
\subsection{Construction} \label{const:kernel}
Let $\mathcal{I}=\{(\mathcal{U}_{1},\mathcal{X}_{1}),\ldots,(\mathcal{U}_{t},\mathcal{X}_{t})\}$ be a collection of $t$ instances of \textsc{Exact-3-Cover}. Without loss of generality, let $\mathcal{U}_{i}=\mathcal{U}=[n]$ and $|\mathcal{X}_{i}|=m$ for all $i\in[t]$ (note that $n=3c$ for some $c\in \mathbb{N}$). Let $\mathcal{C}=\displaystyle{\bigcup_{i\in[t]}\mathcal{X}_{i}}$, and assume that $\mathcal{X}_{i}\neq \mathcal{X}_{j}$ for every distinct $i,j\in[t]$. Also, we denote $\mathcal{C}$ as $\{c_{1},c_{2},\ldots, c_{|\mathcal{C}|}\}$. Furthermore, we denote by $(G,\ell)$ the instance of \textsc{Uniquely Restricted Matching} that we construct in this section.  First, we introduce the vertex set $\widehat{U}=\displaystyle{\bigcup_{a\in \mathcal{U}}}\{v_{a},v'_{a},u_{a},u'_{a}\}$ and edge set $\displaystyle{\bigcup_{a\in \mathcal{U}}}\{v_{a}v'_{a},v'_{a}u'_{a},u'_{a}u_{a},u_{a}v_{a}\}$ to $G$. 

Next, we define the set gadget as follows.
\medskip
\medskip

\noindent \textbf{Set Gadget:} For each $c_{j}\in \mathcal{C}$, where $c_{j}=\{a,b,c\}$, we introduce the following vertex set $\{p_{j},q_{j},r_{j},s_{j},$ $w_{ja}, w_{jb},w_{jc},w_{j1},w_{j2},w'_{ja},w'_{jb},w'_{jc},w'_{j1},w'_{j2}\}$. We also introduce the following edge set $\{w_{j1}w_{j2},w_{j1}w_{ja},$ $w_{j1}w_{jb},w_{j1}w_{jc},w_{j2}p_{j},w_{ja}p_{j},w_{jb}p_{j},w_{jc}p_{j},w'_{j1}w'_{j2},w'_{j1}w'_{ja},w'_{j1}w'_{jb},w'_{j1}w'_{jc},w'_{j2}s_{j}, $ $w'_{ja}s_{j},$   $w'_{jb}s_{j},$ $w'_{jc}s_{j},p_{j}q_{j},$ $q_{j}r_{j}$ $,s_{j}r_{j},p_{j}r_{j}\}$ to $G$. Let us refer to the set gadget corresponding to $c_{j}\in \mathcal{C}$ as $Q_{j}$. Also, for every $c_{j}\in \mathcal{C}$, let the subgraph of $Q_{j}$ induced by the set $\{p_{j},q_{j},r_{j},s_{j}\}$ be denoted by $Q'_{j}$. See Fig.~\ref{fig3} for an illustration. 

\medskip

\begin{figure}[h]
 \centering
    \includegraphics[scale=0.7]{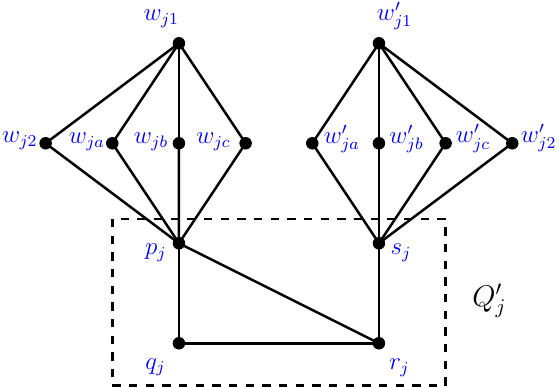}
    \caption{The construction of a set gadget $Q_{j}$ corresponding to $c_{j}=\{a,b,c\}\in  \mathcal{C}$.}
    \label{fig3}
\end{figure}

Furthermore, for every $Q_{j}$, if $c_{j}=\{a,b,c\}$, then we call the vertices $\{w_{ja},w_{jb},w_{jc}\}$ \emph{first set of interface vertices} and call the vertices $\{w'_{ja},w'_{jb},w'_{jc}\}$ \emph{second set of interface vertices}. For every $c_{j}\in \mathcal{C}$, where $c_{j}=\{a,b,c\}$ and $v_{d}\in \widehat{U}$, $w_{ja}v_{d}\in E(G)$ and $w'_{ja}v'_{d}\in E(G)$ if and only if $a=d$. For every distinct pair $c_{j},c_{k}\in \mathcal{C}$, if $a\in c_{j},c_{k}$, then introduce edges $w_{ja}w'_{ka},w'_{ja}w_{ka}$ in $G$. 

Next, we define the instance selector as follows.

\smallskip

\noindent \textbf{Instance Selector:} Introduce a $K_{1,t}$ with $y$ as the central vertex and $x_{i}$, $i\in [t]$, as leaves. Let $X=\{x_{i}:i\in [t]\}$. For every $x_{i}, i\in [t]$ and $c_{j}\in \mathcal{C}\setminus \mathcal{X}_{i}$, do the following:
\begin{itemize}
    \item Introduce edges between $x_{i}$ and the first set of interface vertices of the set gadget $Q_{j}$.
    \item Introduce edges between $x_{i}$ and $w_{j2}$ and $y$ and $w'_{j2}$.
    \item  Introduce edges between $y$ and the second set of interface vertices of the set gadget $Q_{j}$.
\end{itemize}   
See Fig.~\ref{fig4} for an illustration of the construction of $G$. 

\begin{figure}[h]
 \centering
    \includegraphics[scale=0.6]{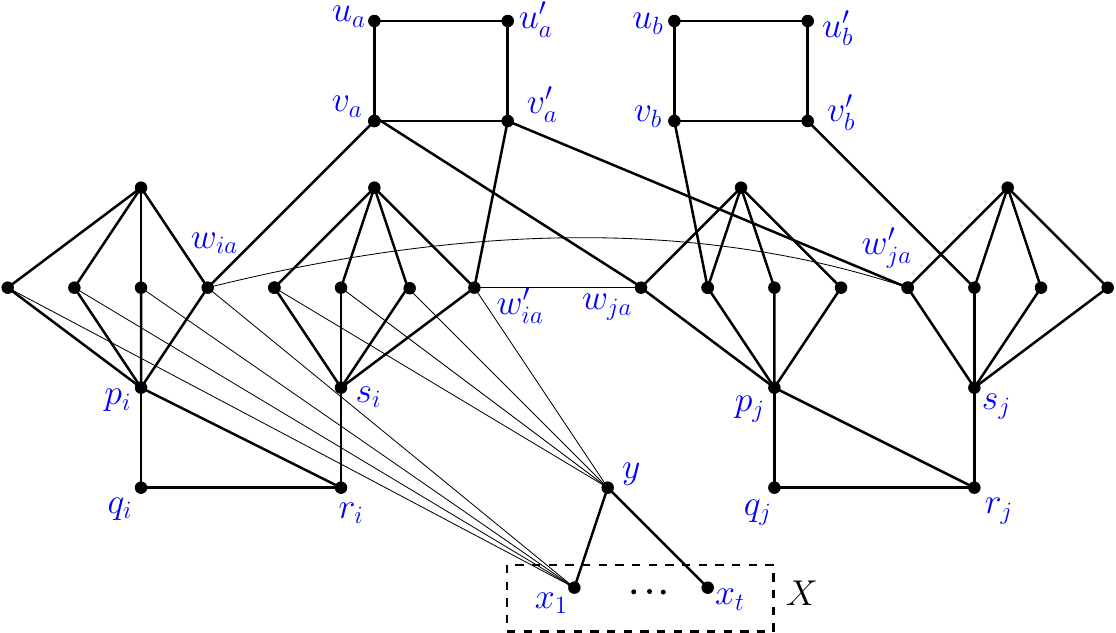}
    \caption{An illustration of the construction of $G$, where, $c_{i}$ contains $a$, $c_{j}$ contains both $a$ and $b$, $c_{i}\notin \mathcal{X}_1$, and $c_{i},c_{j}\in \mathcal{X}_t$. Edges are shown with varying thickness solely to enhance the clarity of the figure.}
    \label{fig4}
\end{figure}

This ends the construction of $G$. Finally, we set $\ell=4|\mathcal{C}|+\frac{8n}{3}+1$.

\subsection{From Exact-3-Cover to Uniquely Restricted Matching}
First, consider the following lemma.
\begin{lemma} \label{sizevc}
Let $G$ be as defined in Subsection \ref{const:kernel}. Then, $G$ has a vertex cover of size $\mathcal{O}(n^{3})$.
\end{lemma}
\begin{proof}
As $X$ is an independent set of $G$, it can be observed that the set $Y=V(G)\setminus X $ is a vertex cover of $G$. Furthermore, note that $|Y|=4n+14|\mathcal{C}|+1$. Since all the elements in $\mathcal{C}$ are distinct, there can be at most $\binom{n}{3}$ elements in $\mathcal{C}$. Thus, we have $|Y|\in \mathcal{O}(n^{3})$. 
\end{proof}

We remark that by saying that $(\mathcal{U},\mathcal{X})$ admits a solution of $\mathcal{I}$, we mean that there exists a set $\mathcal{X'}\subseteq \mathcal{X}$ such that if $\mathcal{X'}=\{c'_{1},\ldots,c'_{\frac{n}{3}}\}$, then $\bigcup^{\frac{n}{3}}_{i=1}$ $c'_{i}=\mathcal{U}$. Now, consider the following lemma.

\begin{lemma} \label{AM-1}
If $(\mathcal{U},\mathcal{X}_{q})$ admits a solution of $\mathcal{I}$, then $G$ admits a uniquely restricted matching of size $\ell$.
\end{lemma}
\begin{proof}
Let $\mathcal{X'}$ be a solution of $(\mathcal{U},\mathcal{X}_{q})$. Now, we will construct a matching $M$ as follows. \begin{itemize}
    \item For every $a\in \mathcal{U}$, add the edge $u_{a}u'_{a}$ to $M$.
    \item For every $c_{j}\in \mathcal{X'}$, if $c_{j}=\{a,b,c\}$, then add the edge set $\{v_{a}w_{ja},v'_{a}w'_{ja},v_{b}w_{jb},v'_{b}w'_{jb},v_{c}w_{jc},v'_{c}w'_{jc}\}$ to $M$.
    \item For every $c_{j}\in \mathcal{C}$, add the edge set $\{w_{j1}w_{j2},w'_{j1}w'_{j2}\}$ to $M$. 
    \item For every $c_{j}\in \mathcal{X'}$, add the edge $q_{j}r_{j}$ to $M$.
    \item For every $c_{j}\in \mathcal{C} \setminus \mathcal{X'}$, add the edge set $\{p_{j}q_{j},r_{j}s_{j}\}$ to $M$. 
    \item Finally, add the edge $x_{q}y$ to $M$. 
\end{itemize}  

Since $|\mathcal{X'}|=\frac{n}{3}$, $|M|=n+2n+2|\mathcal{C}|+\frac{n}{3}+2(|\mathcal{C}|-\frac{n}{3})+1=\ell$. Also, it is immediate to see that $M$ is a matching. However, it remains to show that $M$ is a uniquely restricted matching of $G$.  

For the sake of contradiction, assume that $M$ is not a uniquely restricted matching. Now, by Theorem \ref{defurm},
% \ig{problem with capital letter in ``theorem''}
there must exist an alternating cycle, say $C$, in $G[V_{M}]$. Furthermore, assume that $C$ is an alternating cycle with the fewest edges among all available alternating cycles. Now, we will show that no edge of $M$ can be a part of $C$, and hence $C$ cannot exist in $G[V_{M}]$. 
\medskip

% \begin{figure}[t]
%  \centering
%     \includegraphics[scale=0.8]{URM2.pdf}
%     \caption{}
%     \label{fig:URM2}
% \end{figure}

First, if $u_{a}u'_{a}$ belongs to $C$ for some $a\in \mathcal{U}$, we claim that there must exist another alternating cycle in $G[V_{M}]$ that does not contain $u_{a}u'_{a}$ whose number of edges is less than the one of $C$. Indeed, since $v_{a}v'_{a}\notin M$, the only way $u_au'_{a}\in C$ is if $v_{a}x$ and $v'_{a}y$ are also in $C$ for some $x$ and $y$. Now, we can shorten the cycle $C$ by replacing the part $xv_{a}u_{a}u'_{a}v'_{a}y$ with $xv_{a}
v'_{a}y$. Thus, we can assume wlog that $C$ does not contain any edge of the form $u_{a}u'_{a}$, where $a\in \mathcal{U}$. 

% \begin{figure}[t]
%  \centering
%     \includegraphics[scale=0.6]{URM4.pdf}
%     \caption{}
%     \label{fig:URM4}
% \end{figure}

Now, we claim that $w_{j1}w_{j2}$ does not belong to $C$ for any $c_{j}\in \mathcal{C}$. If $c_{j}\in \mathcal{X'}$, then none of the neighbors of $w_{j2}$ except $w_{j1}$ will be saturated by $M$ as $x_{q}$ is not adjacent to $w_{j2}$ in this case. On the other hand, if $c_{j}\in \mathcal{C}\setminus\mathcal{X'}$, then none of the neighbors of $w_{j1}$ except $w_{j2}$ will be saturated by $M$. 

Now, note that $yx_{q}$ cannot belong to $C$ as for $yx_{q}$ to belong to an alternating cycle, some edge of the form $w_{j1}w_{j2}$ where $c_{j}\in \mathcal{C} \setminus \mathcal{X}_{q}$ would also need to be included in the cycle, which is not possible by the arguments given above. 

Now, we claim that $w'_{j1}w'_{j2}$ does not belong to $C$ for any $c_{j}\in \mathcal{C}$.
Indeed, for $w'_{j1}w'_{j2}$ to be part of an alternating cycle, note that $x_{q}y$ must also be part of that cycle. This is true because if $c_{j}\in \mathcal{X'}$, then the only $M$-saturated neighbor of $w'_{j2}$ other than $w'_{j1}$ is $y$. On the other hand, if $c_{j}\in \mathcal{C}\setminus\mathcal{X'}$, then the only $M$-saturated neighbor of $w'_{j1}$ other than $w'_{j2}$ is $x_{q}$. By the above arguments, we know that $yx_{q}$ does not belong to $C$, and thus $w'_{j1}w'_{j2}$ cannot belong to $C$ for any $c_{j}\in \mathcal{C}$.

% The same arguments hold for $c_{k}\notin \mathcal{C'}$, that is, for either $w_{k1}w_{k2}$ or $w'_{k1}w'_{k2}$ to belong to $C$, the edge $x_{q}y$ must also be a part of that cycle, which cannot be true.

If $v_{a}w_{ja}$ (resp. $v'_{a}w'_{ja}$) belongs to $C$ for some $c_{j}$ that contains $a$, then $C$ must contain $w_{j1}w_{j2}$ (resp. $w'_{j1}w'_{j2}$) as well, as the only $M$-saturated neighbor of $w_{ja}$ (resp. $w'_{ja}$) other than $v_{a}$ (resp. $v'_{a}$) is $w_{j1}$ (resp. $w'_{j1}$). But since we have already proved that edges of the form $w_{j1}w_{j2}$ (resp. $w'_{j1}w'_{j2}$) do not belong to $C$ for any $c_{j}\in \mathcal{C}$, we are done.

If $q_{j}r_{j}$ belongs to $C$ for some $c_{j}\in \mathcal{X}'$, it is not possible because none of the neighbors of $q_{j}$ or $r_{j}$ are $M$-saturated, ruling out such edges from $C$.

For $p_{j}q_{j}$ to belong to $C$ for some $c_{j}\in \mathcal{C}\setminus\mathcal{X}'$, it would require edges of the form $w_{j1}w_{j2}$ to be part of $M$, which is not possible.

Similarly, for $s_{j}r_{j}$ to belong to $C$ for some $c_{j}\in \mathcal{C}\setminus\mathcal{X}'$, edges of the form $w'_{j1}w'_{j2}$ would need to be part of $M$, which is not possible.

Since none of the edges in $M$ can be part of the alternating cycle $C$, it follows that no alternating cycle $C$ exists in $G$, and thus $M$ is a uniquely restricted matching in $G$.
\end{proof}

\subsection{From Uniquely Restricted Matching to Exact-3-Cover} Let $G$ be as defined in Subsection \ref{const:kernel}. Now, consider the following partition of the edges of $G$. See Fig.~\ref{fig5} for an illustration.

\begin{itemize}
    \item Type-I: $\bigcup_{a\in \mathcal{U}}\{v_{a}v'_{a},v'_{a}u'_{a},u'_{a}u_{a},u_{a}v_{a}\}$.
    \item Type-II: For every $a\in\mathcal{U}$ and every $c_{j}\in \mathcal{C}$ such that $a\in c_{j}$, edges of the form $\{v_{a}w_{ja},v'_{a}w'_{ja}\}$. Edges of these kinds are also called  \emph{vertical edges}.
    
    For every $a\in\mathcal{U}$ and for every distinct $c_{j},c_{k}\in \mathcal{C}$ such that $a$ is an element of both  $c_{j},c_{k}$, edges of the form $\{w_{ja}w'_{ka},w_{ka}w'_{ja}\}$. Edges of these kinds are also called   \emph{horizontal edges}.

\item Type-III: For every $c_{j}\in \mathcal{C}$, if $c_{j}=\{a,b,c\}$, then edges of the form  $\{w_{ja}p_{j},w_{jb}p_{j},w_{jc}p_{j},$ $w'_{ja}s_{j},w'_{jb}s_{j},\\w'_{jc}s_{j},w_{ja}w_{j1},w_{jb}w_{j1},w_{jc}w_{j1},w'_{ja}w'_{j1},w'_{jb}w'_{j1},w'_{jc}w'_{j1},w_{j1}w_{j2},w'_{j1}w'_{j2},w_{j2}p_{j},w'_{j2}s_{j}\}$. Furthermore, for every $c_{j}\in \mathcal{C}\setminus \mathcal{X}_{i},$  edges of the form $\{x_{i}w_{j2},x_{i}w_{ja},x_{i}w_{jb},x_{i}w_{jc},yw'_{j2},yw'_{ja},yw'_{jb},yw'_{jc}\}$.

\item Type-IV: For every $c_{j}\in \mathcal{C}$, edges of the form $\{p_{j}q_{j},q_{j}r_{j},r_{j}s_{j},p_{j}r_{j}\}$.

\item Type-V: Edges of the form $\{x_{i}y:i\in [t]\}$.

\end{itemize}

\begin{figure}[h]
 \centering
    \includegraphics[scale=0.6]{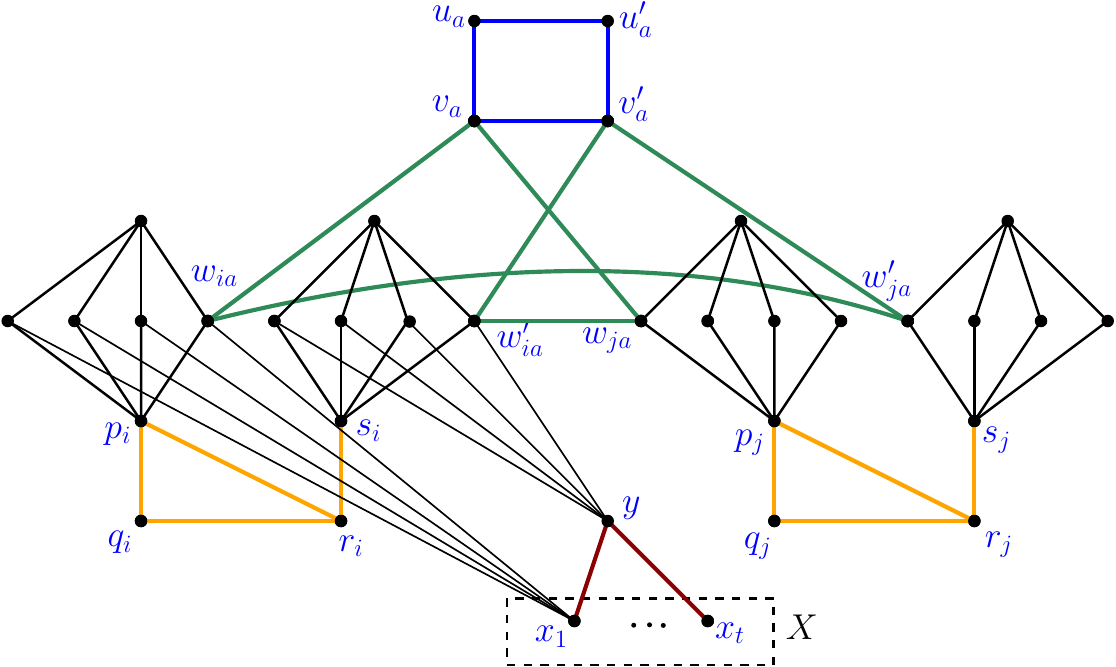}
    \caption{An illustration of the construction of $G$, where both $c_{i}$ and $c_{j}$ contain $a$. Here, Type-I edges are shown in blue, Type-II in green, Type-III in black, Type-IV in orange, and Type-V in deep red.}
    \label{fig5}
\end{figure}

 Next, in Lemmas~\ref{lem:1}-Lemma\ref{lem:5}, we derive an upper bound for each edge type on the number of its edges that can appear in $M$.
% \begin{figure}[t]
%  \centering
%     \includegraphics[scale=0.6]{URM5.pdf}
%     \caption{}
%     \label{fig:URM5}
% \end{figure}

\begin{lemma} \label{lem:1}
For any urm $M$ of $G$, $|M\cap Type$-$I|\leq n$.   
% Furthermore, if $M$ is a urm of $G$, then there exists a urm $M'$ of $G$ such that for every $a\in \widehat{U}$, $u_{a}u'_{a}\in M'$ and $|M'|\geq |M|$.
\end{lemma}
\begin{proof}
  Let $M$ be a urm of $G$. Since Type-I edges consist of a disjoint union of $n$ cycles of length $4$, by Proposition \ref{defurm} at most one edge from each of these $n$ cycles can be part of a urm. Thus, $|M\cap Type$-$I|\leq n$.
\end{proof}

\begin{lemma} \label{lem:2}
For any urm $M$ of $G$, $|M\cap Type$-$III|\leq 2|\mathcal{C}|$.
\end{lemma}
\begin{proof}
  Let $M$ be a urm of $G$. For some arbitrary $c_{j}=\{a,b,c\}$, let $E_{j}$ denote the set of edges $\{w_{j1}w_{j2},w_{j1}w_{ja},$ $w_{j1}w_{jb},$ $w_{j1}w_{jc},p_{j}w_{j2},p_{j}w_{ja},p_{j}w_{jb},p_{j}w_{jc}\}$ and let $E'_{j}$ contain the set of edges $\{w'_{j1}w'_{j2},w'_{j1}w'_{ja},w'_{j1}w'_{jb},$ $w'_{j1}w'_{jc},$ $s_{j}w'_{j2},s_{j}w'_{ja},s_{j}w'_{jb},s_{j}w'_{jc}\}$. Due to Proposition \ref{defurm}, note that $|M\cap E_{j}|\leq 1$, and similarly, $|M\cap E'_{j}|\leq 1$. Also, note that if for some $x_{i}\in X$, $x_{i}w_{jx}\in M$, where $x\in \{a,b,c,2\}$, then $M\cap E_{j}=\emptyset$. Else, an alternating cycle of the form $w_{jx},x_{i},w_{jx'},p_{j},w_{jx}$ or $w_{jx},x_{i},w_{jx'},w_{j1},w_{jx}$ will be formed in $G[V_{M}]$ where $x'(\neq x)\in\{a,b,c,2\}$. Similarly, if $yw'_{jx}\in M$, where $x\in \{a,b,c,2\}$, then $M\cap E'_{j}=\emptyset$. Thus, by the definition of Type-III edges, it follows that $|M\cap Type$-$III|\leq 2|\mathcal{C}|$.\end{proof}

\begin{lemma}\label{lem:3}
For any urm $M$ of $G$, $|M\cap Type$-$IV|\leq 2|\mathcal{C}|$.
\end{lemma}
\begin{proof}
  Let $M$ be a urm of $G$. The lemma holds by the fact that for every $c_{j}\in \mathcal{C}$, $|M\cap E(Q'_{j})|\leq 2$. 
\end{proof}
\begin{remark} \label{rmrk:1}
    If $M$ is a matching of $G$, then for every $c_{j}\in \mathcal{C}$, the following hold: $|M\cap E(Q'_{j})|= 2$ if and only if $p_{j}q_{j},r_{j}s_{j}\in M$. 
\end{remark}
\begin{lemma} \label{lem:4}
For any matching $M$ of $G$,  $|M\cap Type$-$V|\leq 1$.
\end{lemma}

\begin{proof} The proof is immediate, by the definition of Type-V edges.\end{proof}
\begin{remark}
   Observe that each Type-II edge can be uniquely associated with an element of $\mathcal{U}$, determined by its endpoints. We refer to this associated element as the \emph{index} used by the Type-II edge. For example, the edge $v_{a}w_{ja}$ is said to use the index $a$. 
\end{remark}
% \begin{remark}
% For a matching $M$ of $G$, we say that a clause gadget $Q_{j}$ is \emph{touched} by an index $a$ if $a\in c_{j}$ and $w_{ja}$ or $w'_{ja}$ is saturated by $M$.
% \end{remark}
\begin{lemma} \label{lem:5}
For any urm $M$ of $G$, each index is used by at most two edges of Type-II. Furthermore,  $|M\cap Type$-$II|\leq 2n$. 
\end{lemma}
\begin{proof}
  Let $M$ be a urm of $G$. First, we show that for any $a\in \mathcal{U}$, if both $v_{a}$ and $v'_{a}$ are saturated by $M$ and $v_{a}v'_{a}\notin M$, then there exists some $c_{j}\in \mathcal{C}$ with $a\in c_{j}$ such that the edges $v_{a}w_{ja}$ and $v'_{a}w'_{ja}$ are in $M$. In other words, the following situation cannot occur: For distinct $c_j,c_k\in \mathcal{C}$ with $a\in c_{j},c_{k}$, $M$ contains the edges
$v_{a}w_{ja}$ and $v'_{a}w'_{ka}$. This is because it would create an alternating cycle in $G[V_{M}]$ of the form $v_{a},w_{ja},w'_{ka},v'_{a},v_{a}$. Thus, if two vertical edges corresponding to the same index belong to $M$, they must saturate vertices in exactly one set gadget.

Next, we claim that if $|M\cap Type$-$II|> 2n$ for some urm $M$, then there must exist an index, say $a$, that appears in at least three Type-II edges. To see why, note that if every index is used in at most two Type-II edges, then since there are at most $n$ indices and every Type-II edge uses exactly one index, it would follow that $|M\cap Type$-$II|\leq 2n$, contradicting the assumption.

Next, assume that only the horizontal edges of Type-II corresponding to the index $a$ are included in $M$. Now, the following cases are possible. First, if $w_{ia}w'_{ja}, w'_{ia}w_{ka}, w_{ja}w'_{ka}$  $(\text{or}\ w_{ja}w'_{la})\in M$, then a cycle of the form $w_{ia},w'_{ja}, w_{ka},w'_{ia},w_{ja},w'_{ka},w_{ia} $ $(\text{or}\ w_{ia},w'_{ja}, w_{ka},w'_{ia},w_{ja},w'_{la},w_{ia})$ is formed in $G[V_{M}]$. Next, if $w_{ia}w'_{ja}, w'_{ia}w_{ja}\in M$, then the third horizontal edge must be of the form $w_{ka}w'_{la}$ where $a\in c_{k},c_{l}$. In this case, a cycle of the form $w_{ia},w'_{ja},w_{ka},w'_{la},w_{ia}$ is formed in $G[V_{M}]$. Lastly, if $w_{ia}w'_{ja},w_{ka}w'_{la}\in M$, then $w_{ia},w'_{ja},w'_{la},w_{ka},w_{ia}$ is an alternating cycle in $G[V_{M}]$, which is a contradiction to the fact that $M$ is a urm. 

Moving forward, assume that $M$ contains both vertical and horizontal Type-II edges corresponding to the index~$a$. First, let there be two vertical edges and at least one horizontal edge in $M$. In this case, first note that the vertical edges must be of the form $v_{a}w_{ja},v'_{a}w'_{ja}$ for some $c_{j}\in \mathcal{C}$ with $a\in c_{j}$. Let the horizontal edge be of the form $w_{ka}w'_{ia}$ for some $c_{k},c_{i}\in \mathcal{C}$ with $a\in c_{k},c_{i}$. Now, an alternating cycle of the form $w_{ka},w'_{ja},v'_{a},v_{a},w_{ja},w'_{ia},w_{ka}$ will be formed in $G[V_{M}]$, which is a contradiction to the fact that $M$ is a urm. 

To this end, assume that $M$ contains at least one vertical edge, say $v_aw_{ja}$ wlog, and at least two horizontal edges. Among the horizontal edges in $M$, at least one must saturate vertices in two set gadgets, say $c_{k},c_{i}\in \mathcal{C}$, that is different from the one that is touched by the vertical edge in $M$, which is $c_{j}$. In this case, an alternating cycle of the form $v_{a},w_{ja},w'_{ka},w_{ia},v_{a}$ (or $v_{a},w_{ja},w'_{ia},w_{ka},v_{a}$) will be formed in $G[V_{M}]$. 

Thus, by the arguments given above, we can say that every index can be used by at most two edges of Type-II and thus, $|M\cap Type$-$II|\leq 2n$.
\end{proof}

Next, to prove the main lemma (Lemma~\ref{AM-2}), we establish a lower bound on the number of Type-II and Type-V edges that can be included in $M$, assuming $M$ is a solution of $(G, \ell)$  as described in Lemmas~\ref{lem:7} and~\ref{am}.  In this process, we invoke Lemma~\ref{lem:6} and introduce the concepts of \emph{happy} and \emph{sad} set gadgets (see Definition~\ref{def:16}).

\begin{lemma} \label{lem:6}
    If the same index, say $a$,  is used for two Type-II edges in $M$, then there exists one set gadget, say $Q_{j}$, such that  $|E(Q'_{j})\cap M|\leq 1$ and $a\in c_{j}$. 
\end{lemma}
\begin{proof}
  Let $M$ be a urm of $G$. To prove the lemma, suppose that index $a\in \mathcal{U}$ is used twice. Now, the following cases are possible. \smallskip

\noindent \textbf{Case 1:} $v_{a}w_{ja},v'_{a}w'_{ja}\in M$ for some $ c_{j}\in \mathcal{C}$ such that $a\in c_{j}$. In this case, if $|E(Q'_{j})\cap M|=2$, then by Remark~\ref{rmrk:1}, $p_{j}q_{j},r_{j}s_{j}\in M$. This implies that there exists an alternating cycle of the form $v_{a},v'_{a},w'_{ja},s_{j},r_{j},q_{j},p_{j},w_{ja},v_{a}$ in $G[V_{M}]$.
\smallskip

\noindent \textbf{Case 2:} $w_{ja}w'_{ka},w_{ka}w'_{ja}\in M$ for some $ c_{j},c_{k}\in \mathcal{C}$ such that $a\in c_{j},c_{k}$. In this case, if $|E(Q'_{j})\cap M|=|E(Q'_{k})\cap M|=2$, then by Remark \ref{rmrk:1}, $p_{j}q_{j},r_{j}s_{j},p_{k}q_{k},r_{k}s_{k}\in M$. This implies that there exists an alternating cycle of the form $w_{ja},p_{j},q_{j},s_{j},r_{j},w'_{ja},w_{ka},p_{k},q_{k},s_{k},r_{k},w'_{ka},w_{ja}$ in $G[V_{M}]$.\smallskip

\noindent \textbf{Case 3:} $v_{a}w_{ja},w_{ka}w'_{ja}\in M$ for some $ c_{j},c_{k}\in \mathcal{C}$ such that $a\in c_{j},c_{k}$ (the case is analogous when $v'_{a}w'_{ja},w'_{ka}w_{ja}\in M$). In this case, if $|E(Q'_{j})\cap M|=2$, then by Remark \ref{rmrk:1}, $p_{j}q_{j},r_{j}s_{j}\in M$. This implies that there exists an alternating cycle of the form $v_{a},w_{ja},p_{j},q_{j},r_{j},s_{j},w'_{ja},w_{ka},v_{a}$ in $G[V_{M}]$.\smallskip

\noindent \textbf{Case 4:} $w_{ia}w'_{ja},w_{ja}w'_{ka}\in M$ for some $ c_{i},c_{j},c_{k}\in \mathcal{C}$ such that $a\in c_{i},c_{j},c_{k}$. In this case, if $|E(Q_{j})\cap M|=2$, then by Remark \ref{rmrk:1}, $p_{j}q_{j},r_{j}s_{j}\in M$. This implies that there exists an alternating cycle of the form $w_{ia},w'_{ja},s_{j},r_{j},q_{j},p_{j},w_{ja},w'_{ka},w_{ia}$ in $G[V_{M}]$.

In all the above cases, whenever for any Type-II edge in the matching corresponding to some index in $\mathcal{U}$, there exist two edges in $M$, then there exists a set gadget which cannot have two edges in $M$.
\end{proof}
% \begin{lemma}
%     If $M$ is a solution of $(G,\ell)$, then $|M\cap Type$-$II|=2n$.
% \end{lemma}
\begin{definition} \label{def:16}
    A set gadget $Q_{j}$ is \emph{sad}  if $|M\cap E(Q'_{j})|\leq 1$; otherwise it is said to be \emph{happy}.
\end{definition}
\begin{lemma} \label{lem:7}
If $M$ is a solution of $(G,\ell)$, then $|M\cap Type$-$II|= 2n$ and exactly $\frac{n}{3}$ set gadgets are sad.
\end{lemma}
\begin{proof}
      Let $M$ be a urm of $G$. Let us assume that $p$ indices are used in at most one Type-II edge in $M$ and $q$ indices are used in exactly two Type-II edges in $M$.

      By Lemma \ref{lem:6}, at least $q$ set gadgets, though not necessarily distinct, must be sad. Since each set gadget has exactly three indices, at least ${q}_{3}$ set gadgets must be distinct.

      To maximize $|M|$, we assume that at least $\frac{q}{3}$ set gadgets
    contribute one edge each, while $2(|\mathcal{C}|-\frac{q}{3})$ set gadgets contribute two edges each to $M$. Now, by Lemmas \ref{lem:1}-\ref{lem:5}, we have $|M|\leq n+p+2q+2|\mathcal{C}|+1+2(|\mathcal{C}|-\frac{q}{3})+\frac{q}{3}=n+p+4|\mathcal{C}|+1+q+\frac{2q}{3}\leq\frac{8n}{3}+4|\mathcal{C}|+1=\ell$. Since we want $|M|\geq \ell$ and $p+q\leq n$ with $q<n$, it follows that $q=n$.

    From the earlier arguments, at least $\frac{n}{3}$ set gadgets must be sad. We now claim that at most $\frac{n}{3}$ set gadgets should be sad. For contradiction, suppose $l>\frac{n}{3}$ gadgets are happy. Then, $|M|\leq n+2|\mathcal{C}|+1+2n+2(|\mathcal{C}|-l)+l=3n+4|C|+1-l<\ell$  for $-l<-\frac{q}{3}$, which is a contradiction.
      
      Therefore, exactly $\frac{n}{3}$ set gadgets are sad.
\end{proof}
\begin{lemma} \label{am}
  If $M$ is a solution of $(G,\ell)$, then $|M\cap Type$-$V|=1$.
\end{lemma}
\begin{proof}
If $M\cap Type$-$V=\emptyset$, then by Lemmas \ref{lem:1}-\ref{lem:7}, $|M|\leq n+2|\mathcal{C}|+2n+2(|\mathcal{C}|-\frac{n}{3})+\frac{n}{3}=\frac{8n}{3}+4|C|<\ell$, which is a contradiction to the fact that $|M|\geq \ell$.
\end{proof}

\begin{lemma} \label{AM-2}
If $(G,\ell)$ admits a solution, then at least one instance $(X,\mathcal{S}_{q})\in \mathcal{I} $ also admits a solution.
\end{lemma}
\begin{proof}
    By Lemma \ref{am}, every solution of $(G,\ell)$ must contain an edge of the form $x_{q}y$ for some $q\in [t]$.  Now, suppose $x_q$ is adjacent to a happy gadget, say $Q_j$, corresponding to $ c_{j}\in \mathcal{C}$ such that $a\in c_{j}$. Then, both $p_jq_j$ and $s_jr_j$ would belong to $M$, along with $x_qy$. This would induce a cycle in $G[V_M]$ of the form $p_jq_j,q_jr_j,r_js_j,s_jw'_{ja},w'_{ja}y,yx_q,x_qw_{ja},w_{ja}p_j$. This contradicts the assumption that $M$ is a urm. Thus, $x_q$ cannot be adjacent to any happy gadget.
    
    Now, to complete the proof, we aim to show that the $\frac{n}{3}$ sad gadgets cover all the elements of $X$. This follows from the fact that there exists a subset of $\mathcal{S}_q$ of size $\frac{n}{3}$, where each element of $X$ appears exactly once. To establish this, we prove that every element appears at least once in a sad gadget. Suppose, for contradiction, that some element does not appear. By Lemma \ref{lem:7}, we know that every index should appear twice to satisfy $|M|$. Furthermore, by the previous assumption, these appearances must occur within happy gadgets. This leads to a contradiction to the definition of happy gadgets due to Lemma \ref{lem:6}. 

    Thus, every index appears at most once in sad gadgets. Since there are exactly $\frac{n}{3}$ sad gadgets, they must collectively cover all elements, completing the proof.
\end{proof}

Thus, by Lemmas \ref{AM-1} and \ref{AM-2}, we have the following theorem.
\vcthm*

\bibliographystyle{elsarticle-num}
\bibliography{lagos}

\end{document}